\documentclass[final,twocolumn,twoside]{IEEEtran}

\IEEEoverridecommandlockouts
\newcommand{\ra}[1]{\renewcommand{\arraystretch}{#1}}

\usepackage{amsmath,mathrsfs, enumerate}
\usepackage{amssymb}
\usepackage[all]{xy}
\usepackage{graphicx,subfigure,xspace,bm}
\usepackage{color}
\usepackage{booktabs}
\usepackage{algorithm}
\usepackage[noend]{algorithmic}
\algsetup{indent=2em}

\newcommand{\algostep}[1]{{\small\texttt{#1\!:}}\xspace}
  
\newtheorem{theorem}{Theorem}[section]

\newtheorem{lemma}[theorem]{Lemma}

\newtheorem{remark}[theorem]{Remark}
\newtheorem{problem}[theorem]{Problem}
\newtheorem{example}[theorem]{Example}
\newtheorem{corollary}[theorem]{Corollary}
\newtheorem{proposition}[theorem]{Proposition}  

\newcommand{\real}{{\mathbb{R}}}

\newcommand{\integers}{\mathbb{Z}}

\newcommand{\oprocendsymbol}{\hbox{$\bullet$}}
\newcommand{\oprocend}{\relax\ifmmode\else\unskip\hfill\fi\oprocendsymbol}

\newcommand\subscr[2]{#1_{\textup{#2}}}

\newcommand{\Nn}{\mathcal{N}}
\newcommand{\argmax}{\mathrm{argmax}}

\newcommand{\G}{{\mathcal{G}}}

\newcommand{\I}{\mathcal{I}}

\renewcommand{\theenumi}{(\roman{enumi})}

\renewcommand{\baselinestretch}{1}

\newcommand{\myclearpage}{\clearpage}
 \renewcommand{\myclearpage}{}

\begin{document}

\title{Distributed Submodular Maximization \\ with Limited Information}

\author{Bahman Gharesifard  \qquad  Stephen L. Smith \thanks{Bahman Gharesifard is with the Department of Mathematics and Statistics at Queen's University, Kingston, ON, Canada~\texttt{ghareb@queensu.ca}.  Stephen L.\ Smith is with the Department of Electrical and Computer Engineering at the University of Waterloo, Waterloo, ON, Canada~\texttt{stephen.smith@uwaterloo.ca}.}\thanks{The work of the authors was partially supported by the Natural Sciences and Engineering Research Council of Canada. A preliminary version of this paper appeared in the Proceedings of the $2016$ American Control Conference~\cite{Gharesifard2016Distributed}.}}

\maketitle

\begin{abstract}
We consider a class of distributed submodular maximization problems in which each agent must choose a single strategy from its strategy set.  The global objective is to maximize a submodular function of the strategies chosen by each agent.  When choosing a strategy, each agent has access to only a limited number of other agents' choices.  For each of its strategies, an agent can evaluate its marginal contribution to the global objective given its information. The main objective is to investigate how this limitation of information about the strategies chosen by other agents affects the performance when agents make choices according to a local greedy algorithm. In particular, we provide lower bounds on the performance of greedy algorithms for submodular maximization, which depend on the clique number of a graph that captures the information structure. We also characterize graph-theoretic upper bounds in terms of the chromatic number of the graph. Finally, we demonstrate how certain graph properties limit the performance of the greedy algorithm. Simulations on several common models for random networks demonstrate our results. 
\end{abstract}

\myclearpage

\section{Introduction}\label{section:intro}

There has been a recent surge of activity in developing efficient algorithms for large-scale optimization problems, by partitioning the objective into smaller subtasks, each of which is optimized separately. A well-studied class of such problems is the distributed optimization problem, where the main objective is to optimize a sum of functions, each one available to an individual agent, by means of limited information sharing between individuals~\cite{JNT-DPB-MA:86, MR-RN:04,LX-SB:06, AN-AO:09,PW-MDL:09,AN-AO-PAP:10,BJ-MR-MJ:09,BG-JC:14-tac, AN-AO:15-tac}. Unlike this setting, where strategy sets are often assumed to be continuous sets, here we consider a collaborative scenario, where agents choose a strategy from a private \emph{discrete} strategy set and the goal is to maximize a common objective function defined over the set of strategies chosen by each agent. In this sense, this work can be considered as one possible way to extend the literature on studying distributed optimization with limited information and continuous strategy sets to combinatorial scenarios. We consider situations where the total strategy set is partitioned into subsets, each available only to an individual agent, and the common objective function is additionally assume to be \emph{submodular}.  When choosing a strategy, each agent has access to only a limited number of other agents' choices, and can evaluate its marginal contribution to the global objective given its information. This problem arises in a variety of applications of interest, including problems in distributed sensor coverage~\cite{Marden2017Role, jawaid2015submodularity}, information gathering~\cite{AAA-MIS:04-pipage,jawaid2015informative}, and facility location~\cite{AK-DG:12}.

Building on the early classic work~\cite{GLN-LAW-MLF:78-I, MLF-GLN-LAW:78-II, LL:83, SF:05}, there has been an emerging body of work on sumbodular optimization, mainly due to its wide set of applications to areas of computer science, see for example~\cite{AK-DG:12} and references therein.
It is well known that the submodular maximization problem is NP-hard~\cite{LL:83}, unlike submodular minimization, which can be solved in polynomial time~\cite{LL:83}. However, good approximation algorithms exist; in particular, when the submodular function is normalized and monotone, a simple greedy algorithm, yields a solution that is within a multiplicative factor of $(1-1/\mathrm{e})$ of the optimal. The distributed problem we consider in this paper can be captured as a partition matroid constraint; in particular, there is a bound on the cardinality of the intersection of the solution set and the strategy set of each individual agent. It has been shown~\cite{GLN-LAW-MLF:78-I} that maximizing a normalized and monotone submodular function subject to a matroid constraint can be approximated to within $1/2$ using the simple greedy algorithm. More recently~\cite{calinescu2011maximizing} a randomized algorithm has been proposed for this problem, which yields an approximation of $(1 - 1/\mathrm{e})$. 

The key feature that differentiates the problem under study in this paper from the classic setting of submodular maximization is the limitations on the information available to a decision maker about the strategies chosen by others. Similar to the literature on distributed optimization, in order to achieve any reasonable performance, agents need to share information about their actions. The main objective of this paper is to investigate how much information is required to guarantee a certain performance, and what limitations are put on the performance by the information structure between agents. 
 
It is worth mentioning that parallel computations in submodular optimization has recently generated a lot of interest~\cite{BM-AK-RS-AK:13,BM-AK-RS-AK:14,SP:13-thesis,clark2015scalable, barbosa2016new}.   These works focus primarily on large-scale maximization or minimization problems, solving them by partitioning the problem into subsets, each of which is optimized separately.  The overall solution is aggregated and refined through additional computations by a central node, often in a MapReduce framework.
In our work, however, we envision scenarios where the limitations imposed on individual agents is physical (for example, being able to only estimate data in a neighbourhood). This paper is also somewhat related to the recent work~\cite{Marden2017Role}, where the role of limitation of information in submodular optimization is studied in the context of coverage problem, for the cases where agents have full information or no information. Here, we address the limitations imposed by the information network topology. Finally, part of our work is related to the so-called ``local greedy algorithms'', which is studied in the classical paper~\cite{MLF-GLN-LAW:78-II}, as we describe in details later.

\paragraph*{Statement of Contributions}
We consider a class of distributed submodular maximization problems in which each agent has a strategy set from which it must choose a single strategy.  The global objective is to maximize a submodular function of the set of strategies chosen by each agent.  The group of agents take decisions sequentially, having available to them only limited information about the strategies chosen by previous agents. Each agent can evaluate its marginal contribution to the global objective given its information.  The information structure is cast as a directed acyclic graph (DAG) and our main objective is to characterize the fundamental limitations that information limitations impose on the performance of local greedy algorithms. We first show that the well-known $ \tfrac{1}{2} $ lower bound on the performance of greedy algorithms for submodular maximization can be obtained when the information structure is a complete DAG. We provide a general lower bound on the performance, which depends on the clique number of the graph, and also provide lower bounds for graph topologies with multiple interconnected cliques. We then characterize graph-theoretic upper bounds on the performance of the local greedy algorithm. We tackle this objective by considering two problem statements. In the first one, we characterize fundamental limitations on the performance of the greedy algorithms on a given graph topology in terms of its chromatic number. Our second result demonstrates how achieving a certain guaranteed performance imposes limitations on the topology of the underlying graph. We characterize the gap between these two bounds, and show how they can be used to efficiently compute limitations of distributed greedy algorithms for scenarios with limited information. Case studies on several types of randomly generated graphs demonstrate the relationship between the performance bounds provided and the true performance of the distributed greedy algorithm.  

A preliminary version of this paper appeared as~\cite{Gharesifard2016Distributed}.  Relative to this early work, key contributions include proofs of all results; a treatment of intersecting strategy sets and synchronous agent updates; applications to sensor coverage; and, extensive simulation results for several random graph models.

\section{Preliminaries}\label{sec:prelim}

Many combinatorial optimization problems can be phrased as a \emph{submodular maximization problem}, which can be stated as follows. Consider a base set of elements $E$, and let $2^E$ be the power set of $E$, containing all of its subsets. Then a function $f:2^E \to \real_{\geq 0}$ is \emph{submodular} if it possesses the property of diminishing returns: For all $A \subseteq B \subseteq S$, and for all $x \in S\setminus B$ we have
\[
f(A\cup\{x\}) - f(A) \geq f(B\cup\{x\}) - f(B).  
\]
We refer to $f(A\cup\{x\}) - f(A)$ as the \emph{marginal reward} of $x$ given $A$, and denote
it by $\Delta(x|A)$. This definition is equivalent to stating that for all $X,Y \subseteq S$, we have that 
\begin{equation}\label{eq:submod-equi}
f(X)+f(Y) \geq f(X\cup Y)+f(X\cap Y).  
\end{equation}

In addition to submodularity, we will consider functions that possess two further properties:
\begin{enumerate}
\item \emph{Monotonicity:}  For all $A \subseteq B \subseteq E$, we have
  $f(B) \geq f(A)$; and
\item \emph{Normalization:}  $f(\emptyset) = 0$.  
\end{enumerate}

Given a normalized submodular function $f$ defined over a base set $E$ containing $|E| = n$ elements, along with a positive integer $k <n$, the submodular maximization problem can be stated as
\begin{align}
\label{eq:submod_objective}
\max_{S \subset E} \;\;& f(S) \\
\label{eq:submod_constraint}
\text{subject to} \quad &|S| \leq k.
\end{align}

It is well known that this optimization problem is NP-hard. It contains as a special case
several well known combinatorial optimization problems including \textsc{Max Cut} and
\textsc{Facility Location}. Even when the function is monotone, the problem is still NP-hard,
and contains as a special case \textsc{Max $k$-Cover}.

When the submodular function is normalized and monotone,  a simple
greedy algorithm, in which a solution is built incrementally by adding the element $x$ to $S$
that maximizes $\Delta(x|S)$ yields a solution that is within a multiplicative factor of
$(1-1/\mathrm{e}) \approx 0.63$ of the optimal. In addition, it is shown
in~\cite{UF:98} that \textsc{Max $k$-Cover} cannot be approximated to a factor
better than $(1-1/\mathrm{e})$, unless P $=$ NP. This implies that the greedy algorithm provides the best possible approximation.

In place of the constraint~\eqref{eq:submod_constraint}, one can use a more general form of
constraint known as a matroid. The constraint in~\eqref{eq:submod_constraint} defines the \emph{uniform
matroid} in which $S$ must belong to $\{A \subseteq E \; :\; |A| \leq k\}$. Another commonly used matroid constraint in optimization is the \emph{partition matroid} defined as follows:  Let $E$ be a base set of elements, and partition $E$ into $n$ disjoint sets $E_1,\ldots,E_n$ and let $k_1,\ldots,k_n$ be positive integers. Then allowable sets are members of $\{A \subseteq E \; :\; |A \cap E_i| \leq k_i\}$.  It has been shown~\cite{GLN-LAW-MLF:78-I} that maximizing a normalized and monotone submodular function subject to a matroid constraint can be approximated to within $1/2$ using the simple greedy algorithm.

\section{Problem Statement}\label{sec:problem_statement}

We begin by introducing a centralized version of the submodular maximization problem studied in this paper. 

Consider the following submodular maximization problem. We are given $ n $ disjoint sets $ X_i $, $ i \in \{1,\ldots, n\} $,  $X = \cup_i X_i$, and a submodular monotone and normalized function $f:2^X \to \real_{\geq 0}$.  We wish to solve
\begin{align}
	\label{eq:matroid-1}
  &\max_{S \subseteq X} f(S) \nonumber \\
  & \text{subject to} \nonumber \\
  & |S\cap X_i| \leq 1 \quad \text{for each $i\in\{1,\ldots,n\}$}.
\end{align}
This is a submodular maximization problem over a partition matroid, and thus from Section~\ref{sec:prelim}, the optimal solution can be approximated to
within a multiplicative factor of $1/2$ using the simple greedy algorithm~\cite{GLN-LAW-MLF:78-I}. 

\subsection{Distributed Optimization Problem}

We consider a \emph{collaborative} submodular analogue of~\eqref{eq:matroid-1}. 

\begin{problem}[Distributed submodular maximization]\label{problem}
 We are given $n$ agents, or players $ V=\{1,\ldots, n\} $. Each agent has an \emph{strategy set} $X_i$, and must choose one strategy (or action) $x_i \in X_i$.  We define $X = \cup_i X_i$ and are given a normalized and monotone increasing submodular function $f: 2^X \to \real_{\geq 0}$.  The goal is for the agents to each choose an $x_i \in X_i$ such that $f(x_1,\ldots,x_n)$ is maximized.
\end{problem}

We consider the case where agents choose their strategies in a sequential manner, starting with agent $1$ and ending with agent $n$.  Each agent $i\in V$ has access to a subset of the strategies that have been chosen  by agents $\{1,\ldots,i-1\}$.  This information is encoded in a directed acyclic graph (DAG) $\G = (V,E)$ where all edges $(i,j) \in E$ satisfy $i < j$.  Note that every DAG can be topologically sorted, and thus there always exists a labeling of the vertices of $\G$ in which $i < j$ for all $(i,j) \in E$.  A \emph{complete DAG} is a DAG for which no edge can be added without creating a cycle. A \emph{clique} in $\G $ is a subgraph of $ \G $ that is complete.   The \emph{clique number} of a DAG $\G$ is the number of vertices in its largest clique, and is denoted $\omega(\G)$.

We define the in-neighbors of agent $i$ as
\[
\Nn(i) = \{j\in V \;|\; (j,i) \in E\},
\]
and thus the information available to agent $i$ when choosing its strategy is
\[
X_{\mathrm{in}}(i) = \{x_j \; |\; j \in \Nn(i)\}. 
\]

In this paper we study the performance of the greedy algorithm in which agent $i$ chooses its strategy $x_i$ to maximize its marginal reward relative to its limited information $X_{\mathrm{in}}(i)$:
\begin{equation}
	\label{eq:seq_greedy}
x_i = \argmax_{x\in X_i}\Delta\big(x\;|\; X_{\mathrm{in}}(i)\big).
\end{equation}
In the case where multiple strategies have equal marginal reward, the greedy algorithm selects one maximizer arbitrarily.  

Given a set of strategies $\{x_1,\ldots,x_n\}$ for the $n$ agents, the overall objective can be written as a sum of marginal rewards as
\begin{equation}
	\label{eq:submod_value}
f(x_1,\ldots,x_n) = \sum_{i = 1}^n \Delta(x_i \;| \; \{x_1,\ldots,x_{i-1}\}),
\end{equation}
where $\{x_1,\ldots,x_0\}$ is defined to be the empty set. It is worth emphasizing that individual agents do not require access to the function $ f $ given in~\eqref{eq:submod_value}.  The require only the ability to calculate their own marginal rewards using the limited information they have access to in order to compute~\eqref{eq:seq_greedy}. Note that while $x_i$ that maximizes agent $ i $'s marginal reward relative to $X_{\mathrm{in}} \subseteq \{x_1,\ldots,x_{i-1}\}$, its contribution to the overall objective value is given by the marginal reward relative to  $\{x_1,\ldots,x_{i-1}\}$. The main goal of this paper is to investigate how this lack of information affects the performance of the greedy algorithm.

\begin{remark}[Topological ordering]
	In general, there are multiple topological orderings of a DAG $\G = (V, E)$, each defining a different sequence in which the agents select their strategies.  However, the strategy chosen by each agent will be identical in all valid orderings, since each agent $i$ must select its strategy after all agents in $X_{\mathrm{in}}(i)$. \oprocend
\end{remark}

\begin{example}[Sensor coverage]
\label{ex:sensor_coverage}
As an illustrative example, consider the setting shown in the left panel of  Figure~\ref{fig:disk_selection}.  
\begin{figure*}[tbh]
	\includegraphics[width=0.32\linewidth]{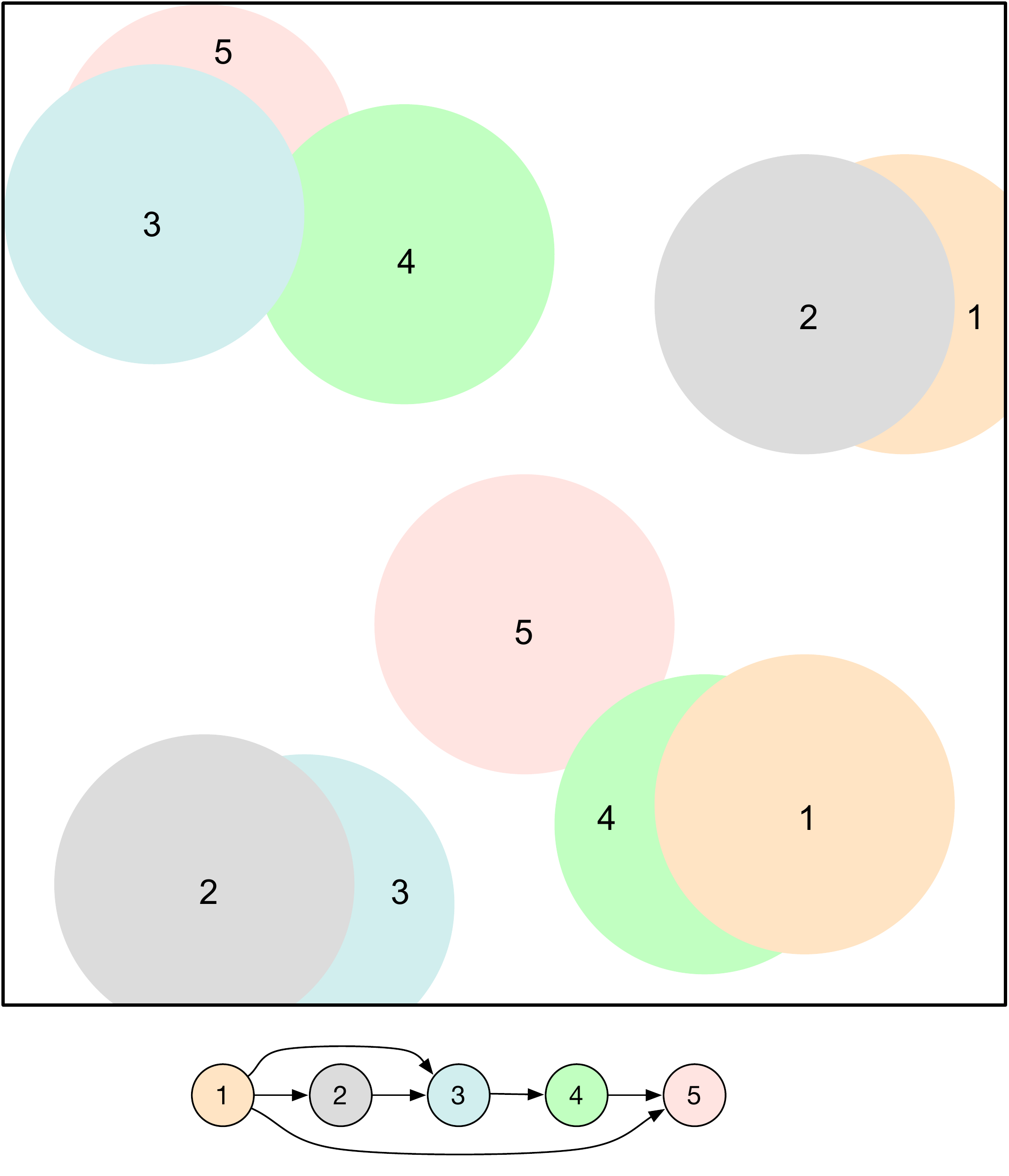} \hfill
	\includegraphics[width=0.32\linewidth]{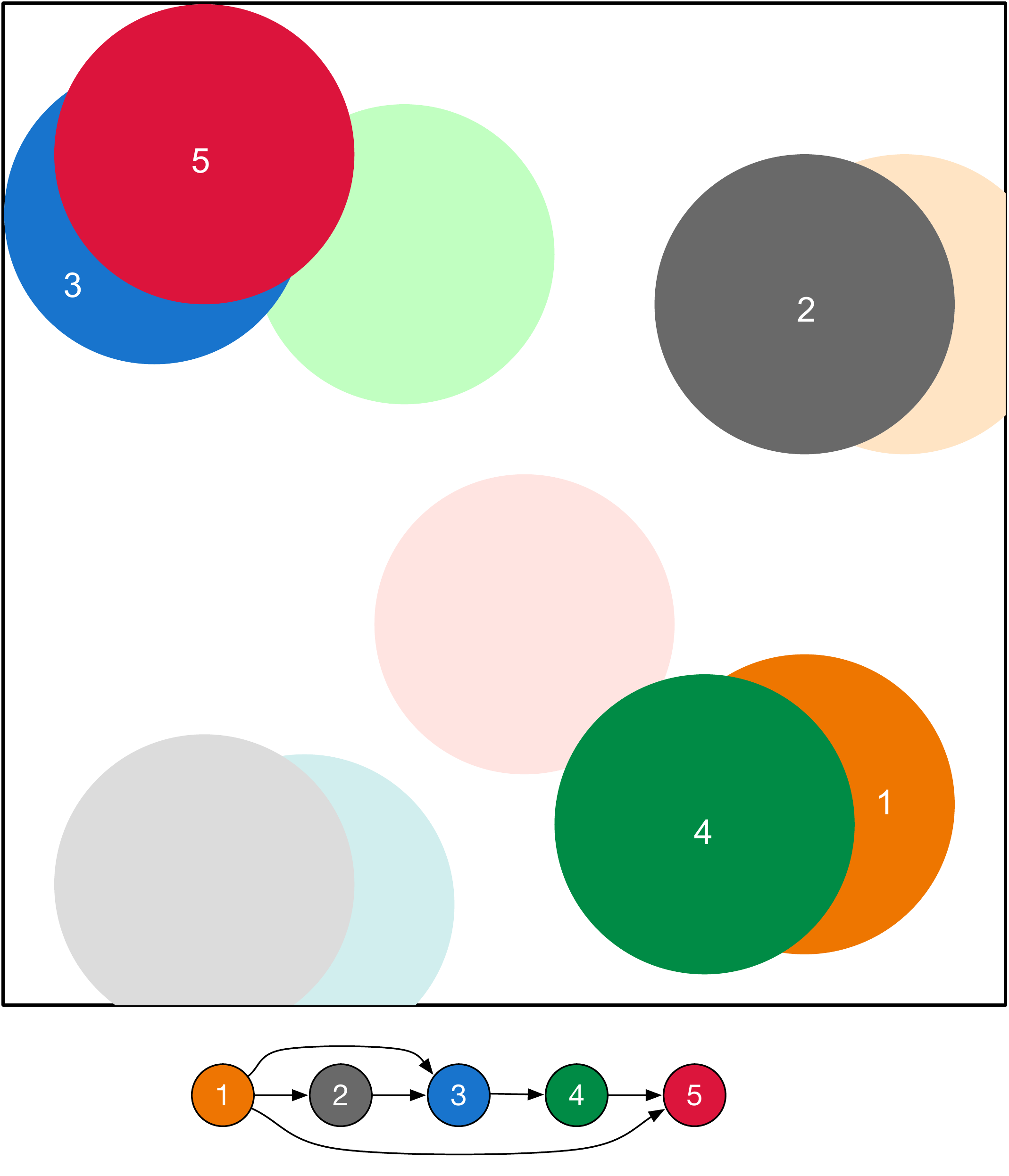} \hfill
	\includegraphics[width=0.32\linewidth]{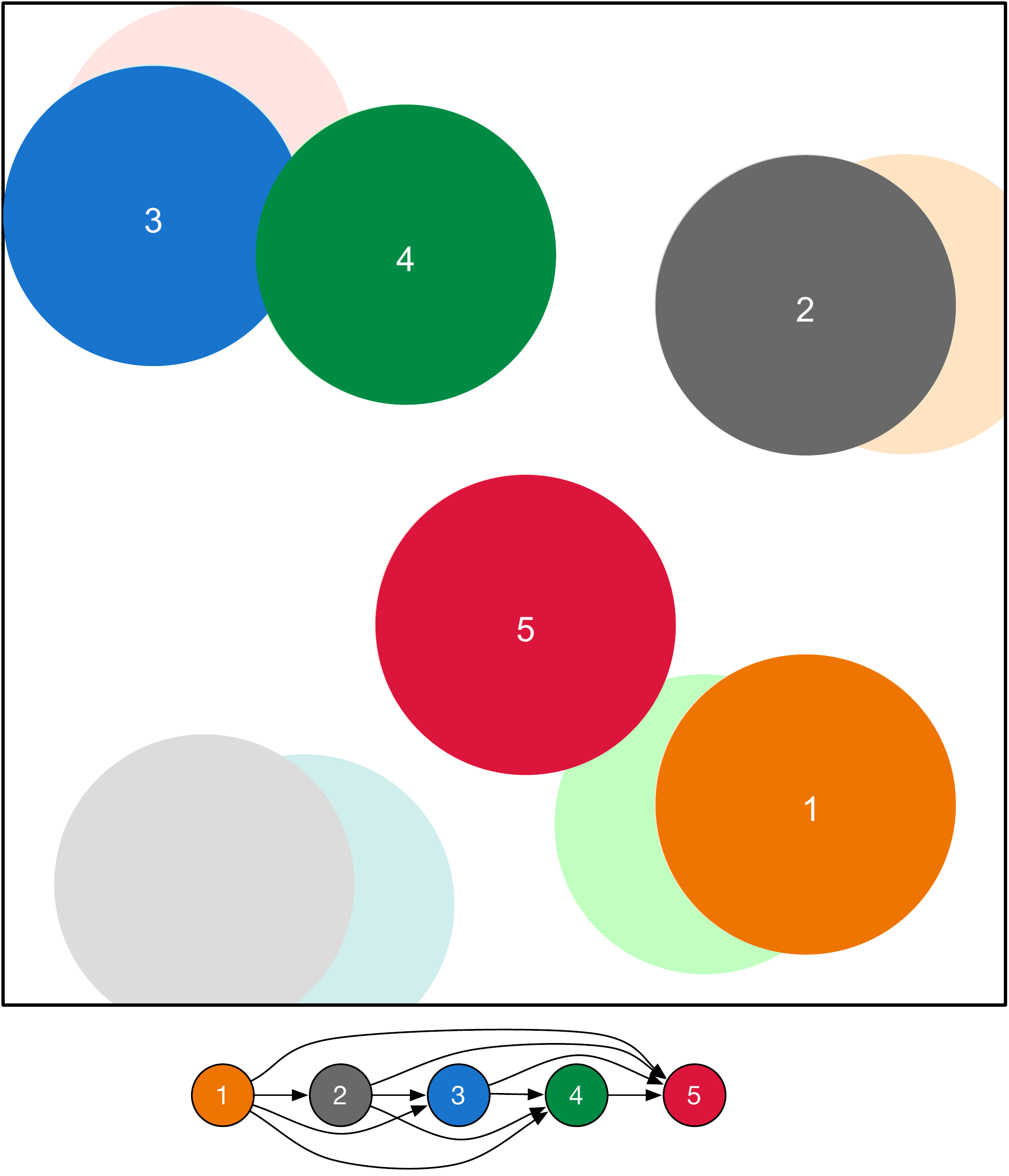}
	\centering
	\caption{Left:  Five agents, each with the choice of two disks. The objective function is the total area covered by the union of disks. Center: The dark numbered disks show the choice of each agent under limited information.  Right:  The dark numbered disks show the choice of each agent given complete information of the choices of prior agents.}
	\label{fig:disk_selection}
\end{figure*}
There are five agents, and each agent $i\in\{1,\ldots, 5\} $ has a strategy set $X_i$ consisting of two elements, which in this example are disks.  Each agent must choose one of its two disks, and thus as an example, agent $2$ must choose one of the two grey disks labeled with the number 2.  The submodular function $f$ gives the total area covered by the union of the disks chosen by each agent.  The information structure $\G$ is shown in the DAG below the panel.  The center panel shows the disk chosen by each agent under the greedy algorithm \eqref{eq:seq_greedy}.  The right panel shows the disks chosen when each agent has complete information of the choices made before it (i.e., when $\G$ is a complete DAG).   We can see that significantly more area is covered when the agents have complete information.  \oprocend 
\end{example}

\begin{remark}[Synchronous agent updates]
	\label{rem:synchronous}
	An alternative model is that each agent $i\in V$ synchronously updates its strategy based on the strategies of its in-neighbors.  That is, let $x_i(t)$ be the strategy of agent $i$ at the synchronous iteration $t$.  Then, agent $i$ greedily chooses $x_i(t)$ to maximize its marginal reward relative to its in-neighbors choices at $t-1$:
	\[
	x_i(t) = \argmax_{x\in X_i}\Delta\big(x\;|\; X_{\mathrm{in}}(i,t-1)\big),
	\] 
	where $X_{\mathrm{in}}(i,t-1) = \{x_j(t-1) \; |\; j \in \Nn(i)\}$.  In this case, it is straightforward to verify that agent $k$'s strategy is fixed after iteration $k$, and its choice of strategy is equal to that in the sequential setup.  Thus, the performance of the synchronous update is identical to that the sequential update. \oprocend
\end{remark}

\subsection{Disjoint vs.\ Intersecting Strategy Sets}
\label{sec:disjoint}

The main results in this paper are bounds on the performance of the greedy algorithm as functions of the information topology $\G$ and the number of agents $n$.  In this section we show that without loss of generality we can restrict our attention to disjoint strategy sets.  A reader can skip directly to Section~\ref{sec:lower_bds} and revisit the formal argument later.

Consider Problem~\ref{problem} with information structure $\G$ and strategy sets $X_1,\ldots,X_n\subseteq X$ that are possibly intersecting.  If the agents run the greedy algorithm in~\eqref{eq:seq_greedy}, then the solution $\bar S = (x_1,\ldots,x_n)$ produced (written as a tuple to keep track of the choice of each agent) has two notable properties: It may contain multiple and/or repeated elements from a single strategy set.  The value of $\bar S$ is evaluated by defining the set $S$ containing the unique elements of $\bar S$, i.e., $S = \{x\in X\;|\; \bar S(i) = x \text{ for some } i\in\{1,\ldots,n\}\}$, and evaluating $f(S)$.  As a simple example, if two agents pick the same strategy $x$, resulting in the solution $\bar S = (x,x)$, then $S = \{x\}$ and the value of the solution is $f(\{x\})$.

Now, we define a corresponding problem with the same number of agents $n$ and information topology $\G$, but with a submodular function $f'$ and disjoint strategy sets $X_1',\ldots,X_n'$ in $X' = \cup_i X_i'$.  Define $X_i'$ as a set containing a unique copy of each element in $X_i$.  We then define a function $g:2^{X'} \to 2^X$ that for a set of elements $A\subseteq X'$, returns the corresponding set of elements in $X$.  Formally, for any set $A \subseteq X'$, we have
\[
g(A) = \{x\in X \;|\;g(a) = x \text{ for some } a\in A\}. 
\]
Note that $g(X_i') = X_i$.  By the definition of $g$, we immediately see that for any $A,B\subseteq X'$, 
\begin{align*}
g(A) \cup g(B) = g(A \cup B), \text{ and } 
g(A) \cap g(B) = g(A \cap B).
\end{align*}

Next, we define the function $f' : 2^{X'} \to \real_{\geq 0}$ as $f' = f\circ g$ so that for any $A\in X'$,
\[
f'(A) = f(g(A)).
\]
To see that the function is submodular, note that for any $A,B\subset X'$, we have
\begin{align*}
f'(A) + f'(B) &= f(g(A)) +f(g(B)) \\ &\leq f(g(A) \cup g(B)) + f(g(A) \cap g(B))\\
&= f(g(A \cup B)) + f(g(A\cap B)) \\
&= f'(A \cup B) + f'(A\cap B),
\end{align*}
where the inequality holds by submodularity of~$f$.  If $f$ is monotone and normalized, then so is $f'$.

Suppose the greedy algorithm~\eqref{eq:seq_greedy} constructs a solution  $S = \{x_1,\ldots,x_n\}$ on the non-disjoint instance $f$.   Consider any $S'=\{x_1',\ldots,x_n'\}$ such that $x_i' \in X_i'$ and $g(x_i) = x_i'$.  What remains is to show that $S'$ is a solution of the greedy algorithm~\eqref{eq:seq_greedy} when run on the disjoint instance $f'$.  Also, since $g(S') = S$ the solutions have equal value, $f'(S') = f(S)$.

By recalling the definition of $X_{\mathrm{in}}(i)$ and defining $X_{\mathrm{in}}'(i) = \{x_j'\;|\; j\in \mathcal{N}(i)\}$, we have
\[
X_{\mathrm{in}}(i) = g(X_{\mathrm{in}}'(i)).
\]
Now, we need to verify that $x_i'$ satisfies
\begin{equation}
\label{eq:xprime_greedy}
x_i' = \argmax_{x'\in X_i'}\Delta'\big(x'\;|\; X_{\mathrm{in}}'(i)\big),
\end{equation}
where $\Delta'$ is the marginal reward for $f'$.  From the definition that $f'(A) = f(g(A))$, we have that
\[
\Delta'\big(x_i'\;|\;X_{\mathrm{in}}'(i) \big) = \Delta\big(x_i\;|\;X_{\mathrm{in}}(i) \big).
\]

Suppose now, by way of contradiction, that $x_i'$ is not a maximizer of~\eqref{eq:xprime_greedy} and there exists a $y'\in X_i'$ such that
\[
\Delta'\big(y'\;|\;X_{\mathrm{in}}'(i) \big) > \Delta'\big(x_i'\;|\;X_{\mathrm{in}}'(i) \big).
\]
This would imply that
$\Delta\big(g(y')\;|\;X_{\mathrm{in}}(i) \big) > \Delta\big(x_i\;|\;X_{\mathrm{in}}(i) \big)$,
and since $g(y') \in X_i$, we conclude that $x_i$ is not a valid choice of the greedy algorithm on $f$, which is a contradiction.  Thus, for every intersecting instance, there is a corresponding disjoint instance on the same information structure, with the same optimal solution value, and such that the greedy algorithm produces solutions with equal value.

\section{The sequential distributed greedy algorithm: lower bound}\label{sec:distributedgreedy}
\label{sec:lower_bds}

Throughout this section, we assume that the strategy sets are disjoint. The agents take their decisions sequentially in increasing order according to their index.  We start with a scenario where the agents \emph{do not observe} the strategies of the agents that have taken action prior to them (see Figure~\ref{fig:seq-obs}(a)). The decision of agent $ i \in V $ is then
\[
x_{i}=\argmax_{x\in X_{i}} \Delta(x | \emptyset).
\]
Suppose that $ \{x_1^*,x_2^*, \ldots, x_{n}^*\} $ be the solution 
of~\eqref{eq:matroid-1}, where $ x^*_i \in X_i $ for each $ i\in \{1,\ldots,n\} $. We have that
\[
f(x_1,\ldots, x_{n}) \geq f(x_i) \geq f(x_i^*),
\]
for all $ i \in \{1,\ldots, n\} $, and hence 
\[
f(x_1,\ldots, x_n) \geq \frac{1}{n} \sum_{i=1}^n f(x^*_i) \geq \frac{1}{n} f(x_1^*,\ldots, x_{n}^*),
\]
where the last inequality follows from submodularity of $ f $ given by~\eqref{eq:submod-equi}.
It is easy to observe that this lower bound is tight and cannot be improved. A natural problem is hence to investigate if this lower bound on the performance can be improved when the agents can observe (perhaps partially) the decisions of the preceding agents before making decisions. Consider the scenario where agent $ i $ observes the decision of all agents in the set $ \I_i=\{1,\ldots, i-1\} $. The information that each agent has access to before taking its decision is best represented by Figure~\ref{fig:seq-obs}(b). In this case, the decision for agent $i$ from~\eqref{eq:seq_greedy} becomes
\begin{equation}\label{eq:asyn-greedy-obs}
x_{i}=\argmax_{x\in X_{i}} \Delta(x | \{x_1,\ldots, x_{i-1}\}).
\end{equation}

\begin{figure}[tbh!]
\centering 
      \subfigure[]{\includegraphics[width=.47\linewidth]{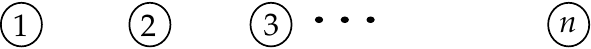}} \quad 
        \subfigure[]{\includegraphics[width=.47\linewidth]{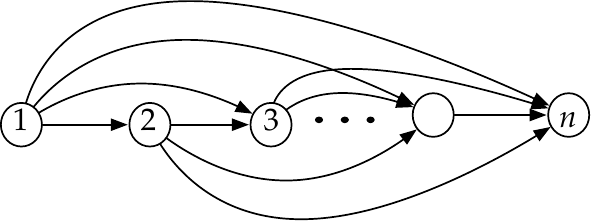}}
\caption{\footnotesize (a)~shows a scenario with no observation and~(b) shows a scenario where each agents observes the decisions of all the preceding agents.}\label{fig:seq-obs}
\end{figure}

The performance of the strategy proposed by~\eqref{eq:asyn-greedy-obs} can be deduced from the so-called ``local greedy algorithms'', which is studied in the classical paper~\cite[Theorem~4.1]{MLF-GLN-LAW:78-II}.  The proof presented in~\cite{MLF-GLN-LAW:78-II} relies on a clever use of linear programming. Here we present an independent proof that allows us to extend this result to more general information topologies. We start with the following result.

\begin{lemma}\label{lemma:inductive-lemma}
Consider Problem~\ref{problem} and let $ (x_1^*,x_2^*, \ldots, x_{n}^*) $ be an optimal solution. Suppose that players sequentially update their strategies $(x_1,\ldots,x_n)$ according to~\eqref{eq:asyn-greedy-obs}. Then, for all $ 1 \leq k \leq n $, we have that
\begin{equation}\label{eq:inductive-lemma}
f(x_1,\ldots, x_k)\geq f(x_1^*,\ldots, x_k^*)-f(x_1,\ldots, x_{k-1}).
\end{equation}
\end{lemma}
\begin{proof}
We prove the following slightly stronger result, 
\begin{align*}
f(x_1, \ldots, x_{k+1}) &\geq f(x_1,\ldots, x_k,x^*_1,\ldots, x^*_{k+1})\\
& - f(x_1,\ldots, x_k).
\end{align*}
Note that the result in~\eqref{eq:inductive-lemma} immediately follows after proving this, since 
\[
f(x_1,\ldots, x_k,x^*_1,\ldots, x^*_{k+1}) \geq f(x^*_1,\ldots, x^*_{k+1}),
\] 
by monotonicity.  The proof will proceed by induction. Throughout the proof, we use  the ``greedy choice property'', i.e., the fact that 
\[
\Delta(x_{k+1}|x_1,\ldots, x_k) \geq \Delta(x^*_{k+1}|x_1,\ldots, x_k). 
\]
First, consider the base case of $k=0$.  By the greedy choice property, we have that
\[
f(x_1) \geq f(x_1^*),
\]
and the the base case holds.  Next, suppose the result holds for $k$ and we show it holds for $k+1$. We have
\begin{align}
f(x_1,\ldots, &x_{k+1})  \nonumber\\
=& f(x_1,\ldots, x_k) + \Delta(x_{k+1}|x_1,\ldots, x_k) \nonumber \\
\geq & f(x_1,\ldots, x_{k-1},x^*_1,\ldots, x^*_{k}) - f(x_1,\ldots, x_{k-1})\nonumber \\
& + \Delta(x_{k+1}^*|x_1,\ldots, x_{k}) \nonumber\\ 
=& f(x_1,\ldots, x_{k-1},x^*_1,\ldots, x^*_{k}) - f(x_1,\ldots, x_{k-1}) \nonumber\\
&+ 	f(x_1,\ldots, x_{k},x_{k+1}^*) - f(x_1,\ldots, x_{k}),
\label{eq:induction_step}
\end{align}
where in the first inequality we applied the inductive hypothesis and the greedy choice property.  Next, using submodularity, we have
\begin{align}
	\label{eq:submod_combine}
f(x_1,\ldots,& x_{k-1},x^*_1,\ldots, x^*_{k}) + f(x_1,\ldots, x_{k},x_{k+1}^*) \\
& \geq f(x_1,\ldots, x_{k},x^*_1,\ldots, x^*_{k+1})\nonumber
 + f(x_1,\ldots, x_{k-1}).
\end{align}
Substituting~\eqref{eq:submod_combine} into the right hand side of~\eqref{eq:induction_step} and cancelling the common $f(x_1,\ldots, x_{k-1})$ term we get
\begin{align*}
f(x_1,\ldots, x_{k+1}) &\geq f(x_1,\ldots, x_{k},x^*_1,\ldots, x^*_{k+1}) \\
&- f(x_1,\ldots, x_{k}),
\end{align*}
which proves the result.
\end{proof}

\begin{theorem}\label{theorem:main-full-one-half}
Consider Problem~\ref{problem} and let $ (x_1^*,x_2^*, \ldots, x_{n}^*) $ be an optimal solution. Suppose that players sequentially update their strategies $(x_1,\ldots,x_n)$ according to~\eqref{eq:asyn-greedy-obs}. Then
\[
f(x_1,\ldots, x_n)\geq \frac{1}{2} f(x^*_1,\ldots, x^*_n).
\]
\end{theorem}
\begin{proof}
Using Lemma~\ref{lemma:inductive-lemma}, we have
\[
f(x_1,\ldots, x_n)\geq f(x_1^*,\ldots, x_n^*)-f(x_1,\ldots, x_{n-1}).
\]
Also, by monotonicity, we have $f(x_1,\ldots, x_n)\geq f(x_1,\ldots, x_{n-1})$, and hence the result follows. 
\end{proof}

The following consequence of this result holds for any general directed acyclic graph.

\begin{corollary}
	\label{cor:clique_bound}
Consider Problem~\ref{problem} and let $ (x_1^*,x_2^*, \ldots, x_{n}^*) $ be an optimal solution. Suppose that players sequentially update their strategies $(x_1,\ldots,x_n)$ according to~\eqref{eq:seq_greedy} on a directed acyclic graph $\G$. Then,
\[
f(x_1,\ldots, x_n)\geq \frac{1}{(n-\omega(\G))+2} f(x^*_1,\ldots, x^*_n),
\]
where $\omega(\G)$ is the clique number of $\G$.
\end{corollary}
\begin{proof}
Let us partition the nodes in $ \G $ into to the set of nodes in a maximum clique, whose strategy sets are denoted by $ \subscr{S}{clique} $ with its complement $ \subscr{\bar{S}}{clique} $. Consider now the smaller optimization problem over $ \subscr{S}{clique} $. By Theorem~\ref{theorem:main-full-one-half}, 
\[
f( \subscr{S}{clique})\geq \frac{1}{2}f( \subscr{S}{clique}^*),
\]
where we have again used the star notation to denote the optimal solution. On the other hand,  the remaining $ n - \omega(\G)$ agents cannot have a performance worst than $ \frac{1}{n-\omega(\G)} $ of their optimal, i.e., 
\[
f(\subscr{\bar{S}}{clique})\geq \frac{1}{n-\omega(\G)}f(\subscr{\bar{S}}{clique}^*).
\]
Using these two inequalities, we conclude that 
\begin{align*}
f( \subscr{S}{clique}^*)+f(\subscr{\bar{S}}{clique}^*)
&\leq 2f( \subscr{S}{clique})+(n-\omega(\G))f(\subscr{\bar{S}}{clique})\\
&\leq (2+(n-\omega(\G))f(S),
\end{align*}
where $ S=\subscr{S}{clique}\cup \subscr{\bar{S}}{clique} $, and the last inequality follows by monotonicity of $ f $. Using the submodularity of $ f $, we have that $ f(S^*) \leq f( \subscr{S}{clique}^*)+f(\subscr{\bar{S}}{clique}^*) $ and hence
\[
 f(S^*)\leq (2+(n-\omega(\G))f(S),
\]
concluding the proof.
\end{proof}

\begin{figure}[htb!]
\centering
	  \includegraphics[width=0.72\linewidth]{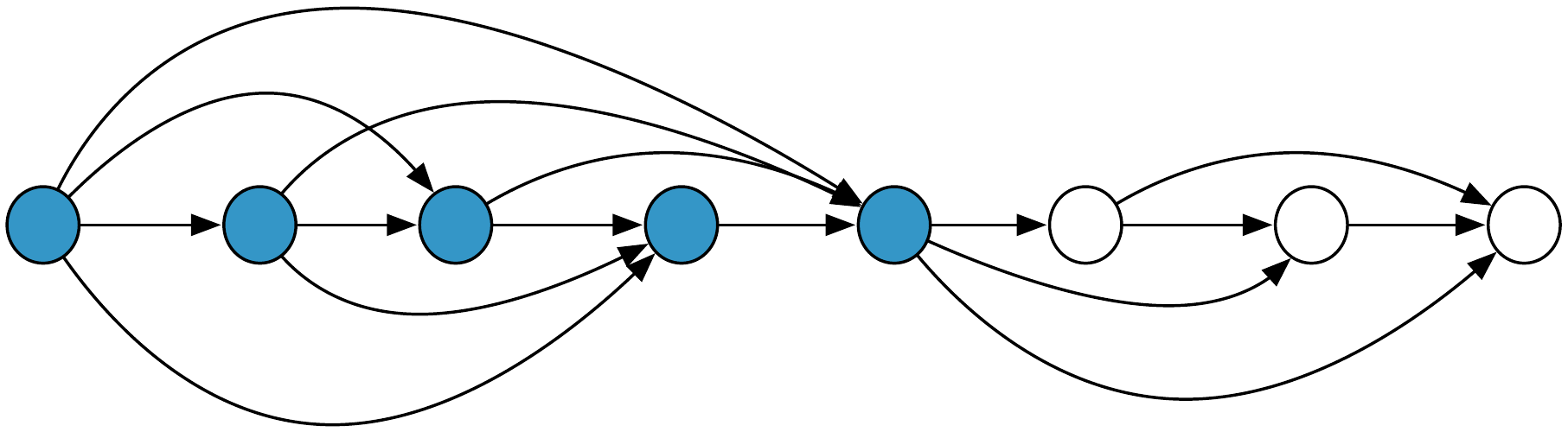}
	  \caption{A DAG with $ 8 $ nodes is depicted, where $ \omega(\G)=5 $. By Corollary~\ref{cor:clique_bound}, we conclude that under~\eqref{eq:seq_greedy}, we have $ f(x_1,\ldots, x_n)>\frac{1}{6}f(x^*_1,\ldots, x^*_n) $.}
	  \label{fig:clique-example}
\end{figure}

Figure~\ref{fig:clique-example} shows an example where Corollary~\ref{cor:clique_bound} is used to conclude the performance of~\eqref{eq:seq_greedy}. Note that if we take a complete DAG and delete a single edge, it's clique number reduces from $n$ to $n-1$, and from Corollary~\ref{cor:clique_bound}, our lower bound decreases from $1/2$ to $1/3$.  The reason for this is that we have not made any assumptions on the relative contribution of each player to the total reward.  The one agent that is removed from the clique may contribute more to the total reward than all $n-1$ other agents combined.  

\subsection{Interconnected Cliques of Full Information}

It is not clear how to characterize the impact of the availability or lack of information on the performance of local greedy algorithms for arbitrary graph topologies. In this sub-section, however, we extend Theorem~\ref{theorem:main-full-one-half} to a scenario with multiple cliques, where the agent's in each clique have access to the decision of the last agent in the last clique that takes decision prior to them.

\begin{theorem}\label{prop:tail-boradcasting-multiple-cliques}
Consider Problem~\ref{problem} and let $ (x_1^*,x_2^*, \ldots, x_{n}^*) $ be an optimal solution. Suppose that agents are partitioned into $ \kappa \in \integers_{\geq 1} $ cliques as
\[
\{1,\ldots, m_1\} , \{m_1+1,\ldots, m_2\}, \cdots \{m_{\kappa-1}+1,\ldots, n\},
\]
with the only information available from other cliques to the agent's in $ (i+1) $'s clique being the choice of agent $ m_i $, $ i\in \{1,\ldots, \kappa-1\} $. Suppose that players sequentially update their strategies $(x_1,\ldots,x_n)$ according to~\eqref{eq:seq_greedy}. Then 
\[
f(x_1,\ldots,x_n)\geq \frac{1}{2\kappa}(f(x^*_1,\ldots,x^*_n)+\sum_{i=1}^{\kappa-1}f(x^*_{m_i})).
\]
\end{theorem}

\begin{proof}
We proceed by induction on $ \kappa $. The base case, where $ \kappa=1 $, follows from Theorem~\ref{theorem:main-full-one-half}. Suppose now that the statement holds for $ \kappa -1 $. We prove that it also holds for $ \kappa$. For notational convenience, let us define $m_{\kappa} := n$.  First, by Theorem~\ref{theorem:main-full-one-half},
\[
f(x_{m_{\kappa-1}},\ldots, x_{m_{\kappa}})\geq  \frac{1}{2}f(\tilde{x}^*_{m_{\kappa-1}},\ldots, \tilde{x}^*_{m_{\kappa}}),
\]
where $ (\tilde{x}^*_{m_{\kappa-1}},\ldots, \tilde{x}^*_{m_{\kappa}}) $ is the optimal strategy for clique $ \kappa $. Hence, we also have that
\[
f(x_{m_{\kappa-1}},\ldots, x_{m_{\kappa}})\geq  \frac{1}{2}f(x^*_{m_{\kappa-1}},\ldots,x^*_{m_{\kappa}}),
\]
where $(x^*_{m_{\kappa-1}},\ldots,x^*_{m_{\kappa}})$ is the strategy for clique $\kappa$ in the true optimal.
By using monotonicity,  we have that 
\[
2f(x_1,\ldots, x_n)\geq  f(x^*_{m_{\kappa-1}},\ldots,x^*_{m_{\kappa}}).
\]
On the other hand, by the induction assumption, along with monotonicity, we have that
\[
2(\kappa -1) f(x_1,\ldots,x_{m_{\kappa-1}})\geq f(x^*_1,\ldots,x^*_{m_{\kappa-1}})+\sum_{i=1}^{\kappa-2}f(x^*_{m_i}).
\]
Summing the last two inequalities, and noting that by monotonicity $f(x_1,\ldots,x_n) \geq f(x_1,\ldots,x_{m_{\kappa-1}})$ since $n = m_{\kappa}$, we have that
\begin{align*}
2\kappa f(x_1,\ldots,x_n)\geq& f(x^*_1,\ldots,x^*_{m_{\kappa-1}})\\
+ f(x^*_{m_{\kappa - 1}},\ldots,x^*_{m_{\kappa}})&+\sum_{i=1}^{\kappa-2}f(x^*_{m_i}).
\end{align*}
Using submodularity, we have that
\begin{align*}
2 \kappa f(x_1,\ldots,x_n)\geq& f(x^*_1,\ldots,x^*_{m_{\kappa}})\\
&+ f(x^*_{m_{\kappa-1}})+\sum_{i=1}^{\kappa-2}f(x^*_{m_i})\\
&=f(x^*_1,\ldots,x^*_n)+\sum_{i=1}^{\kappa-1}f(x^*_{m_i}),\\
\end{align*}
concluding the proof. 
\end{proof}

As noted in Remark~\ref{rem:synchronous}, the results of this section are extend to the case of synchronous updates.
As a final observation, consider the setting of Theorem~\ref{prop:tail-boradcasting-multiple-cliques} and suppose that there is no communications between agents in the cliques. Then, since
\begin{align*}
f(x_1,\ldots, x_{m_1})&\geq \frac{1}{2}f(x^*_1,\ldots, x^*_{m_1}), \quad \mathrm{and} \\
f(x_{m_i+1},\ldots, x_{m_{i+1}})&\geq \frac{1}{2}f(x^*_{m_i+1},\ldots, x^*_{m_{i+1}}),
\end{align*}
for all $ i \in \{1,\ldots, \kappa-1\} $, 
using submodularity and monotonicity, we conclude that $ f(x_1,\ldots, x_n) \geq \frac{1}{2\kappa}f(x^*_1,\ldots, x^*_n) $. In this sense, Theorem~\ref{prop:tail-boradcasting-multiple-cliques} captures a scenario where the overall performance can only be enhanced by having access to more information. In particular, the performance is guaranteed to improve by at least 
\[
\frac{1}{2\kappa} \sum_{i=1}^{\kappa-1}f(x^*_{m_i})\geq \frac{1}{2\kappa} f(x^*_{m_1},\ldots,x^*_{m_{\kappa-1}}) \geq 0,
\]
where the first inequality holds by submodularity.

\myclearpage
\section{Upper bounds using colouring}\label{section:upper}

Our main objective in this section is to determine upper bounds on the performance of the sequential greedy algorithm.  We proceed with studying two cases:  1) we construct a graph dependent submodular function that exploits information on graph topology to yield an upper bound; and 2) we design a global submodular function that provides bounds on all graph topologies. We will see that these two cases are essentially equivalent to computing an optimal and a greedy coloring~\cite{BK-JV:07} of the graph, respectively. The second case also allows us to extract graph properties that limit performance.  

\subsection{Background on Graph Coloring}

Given a graph $G = (V,E)$ with $|V| = n$, a \emph{coloring} is a function $c:V\to \{1,\ldots,n\}$ such that for every $(u,v) \in V$, we have $c(u) \neq c(v)$.  The \emph{vertex-coloring problem} is to find a coloring $c:V \to \{1,\ldots,k\}$ such that $k$ is minimized.  That is, the goal is to color the vertices of the graph with the minimum number of colors such that no pair of adjacent vertices have the same color.  The optimum value of $k$ is called the \emph{chromatic number} of the graph $G$ and is often denoted by $\chi(G)$.  The problems of computing a minimum coloring or equivalently of determining the chromatic number of a graph are NP-hard~\cite{BK-JV:07}.  

\subsection{Upper Bound via Adversarial Choice of Function}
\label{sec:adversarial}

Consider the following problem: Given a graph $\G = (V,E)$, can we construct a submodular function $f_{\G}$, exploiting information of the graph for which we can upper bound the performance of greedy algorithm?  We call such a function \emph{adversarial} since the function can be designed to exploit the weaknesses of the information graph.  Given a graph $\G = (V, E)$, we construct strategy sets and a submodular function as follows.  Let $c:V \to \{1,\ldots,k\}$ be an optimal coloring of $\G$ and where $\chi(\G) = k$ is the chromatic number of $\G$.   Given a coloring, we define the set of vertices of color $\ell \leq k$ as
\[
V_{\ell} = \{v\in V \;|\; c(v) = \ell \}.
\]
We use this coloring to construct strategy sets $X_i$ for each agent $i\in V$ and the submodular function $f: 2^X \to \real_{\geq 0}$.  For each agent $i$, we define 
$X_i = \{a_i, b_i\}$,
and let $X = \cup_i X_i$.  We define the function $f$ through its marginal rewards.  To that end, let $S \subset X$ be a set such that $|S\cap X_i| \leq 1$ containing decision for a subset of the $n$ agents in $V$.  Then, the marginal reward of $a_i$ is defined as
\[
\Delta(a_i | S) = 
\begin{cases}
	0 & \text{if $a_j \in S$ for some $j\in V_{c(i)}$}, \\
	1 & \text{otherwise}.
\end{cases}
\]
The marginal reward of $b_i$ is defined is $\Delta(b_i|S) = 1$, for all $i\in V$.  Then, given a set $S = \{s_1,\ldots,s_m\} \subset X$ with $|S\cap X_i| \leq 1$, for all $ i \in V $, we define
\begin{equation}
	\label{eq:adversarial_submod}
f(S) = \sum_{i=1}^m \Delta\big(s_i|\{s_1,\ldots,s_{i-1}\}\big).
\end{equation}

	Note that we can equivalently define this function over non-disjoint strategy sets.  Let
	$X = \{a,b_1,\ldots,b_n\}$ and define $X_i = \{a,b_i\}$.  Then given a set $S \subset X$, we define $f(S) = |S|$. Notice that this function is modular.  Given a tuple containing agent strategies $(x_1,\ldots,x_n)$, we define the set of unique strategies as $S = \{x_1,\ldots,x_n\}$ and evaluate $f(S)$ as the reward.  This definition is equivalent to the disjoint definition above as discussed in Section~\ref{sec:disjoint}. 

Using the function~\eqref{eq:adversarial_submod}, we obtain the following result.
\begin{proposition}
	\label{prop:coloring_upper}	
	Consider Problem~\ref{problem}, where $f$ is given by~\eqref{eq:adversarial_submod},
	and let $ (x_1^*,x_2^*, \ldots, x_{n}^*) $ be an optimal solution. Suppose that players sequentially update their strategies $(x_1,\ldots,x_n)$ according to~\eqref{eq:seq_greedy}.  Then, 
	\[
	f(x_1,\ldots,x_n) \leq \frac{\chi(G)}{n}f(x_1^*,\ldots,x_n^*).
	\]
\end{proposition}
\begin{proof}
An optimal solution for maximizing~\eqref{eq:adversarial_submod} is $x_i^* = b_i$ for each $i\in V$, which yields a reward of $f(x_1^*,\ldots,x_n^*) = n$.  Now, considering the  greedy algorithm in~\eqref{eq:seq_greedy}, we see that for each agent $i\in V$, we have $\Nn(i) \cap V_{c(i)} = \emptyset$, and thus no agent in $V_{c(i)}$ is visible to agent $i$.  Thus, for each agent $i$ we have
\[
\argmax_{x\in X_i}\Delta\big(x\;|\; X_{\mathrm{in}}(i)\big) = \{a_i, b_i\},
\]
both of which achieve a marginal reward of $1$.  A possible evolution of the greedy algorithm is $x_i = a_i$ for each $i\in V$.  The value of this solution is
\[
f(a_1,\ldots,a_n) = \sum_{i=1}^n \Delta(a_i \; | \; a_1,\ldots,a_{i-1}) = \chi(\G)
\]
since only the first agent of each color contributes $1$ to the total reward, and all others contribute zero.
\end{proof}

Figure~\ref{fig:chromatic-example} depicts the DAG used in Figure~\ref{fig:clique-example} along with an optimal coloring yielding a chromatic number of $ 5 $. Using Proposition~\ref{prop:coloring_upper}, we conclude that under~\eqref{eq:seq_greedy}, we have $ f(x_1,\ldots, x_n)<\frac{5}{8}f(x^*_1,\ldots, x^*_n) $. 

\begin{figure}
\centering
	  \includegraphics[width=0.72\linewidth]{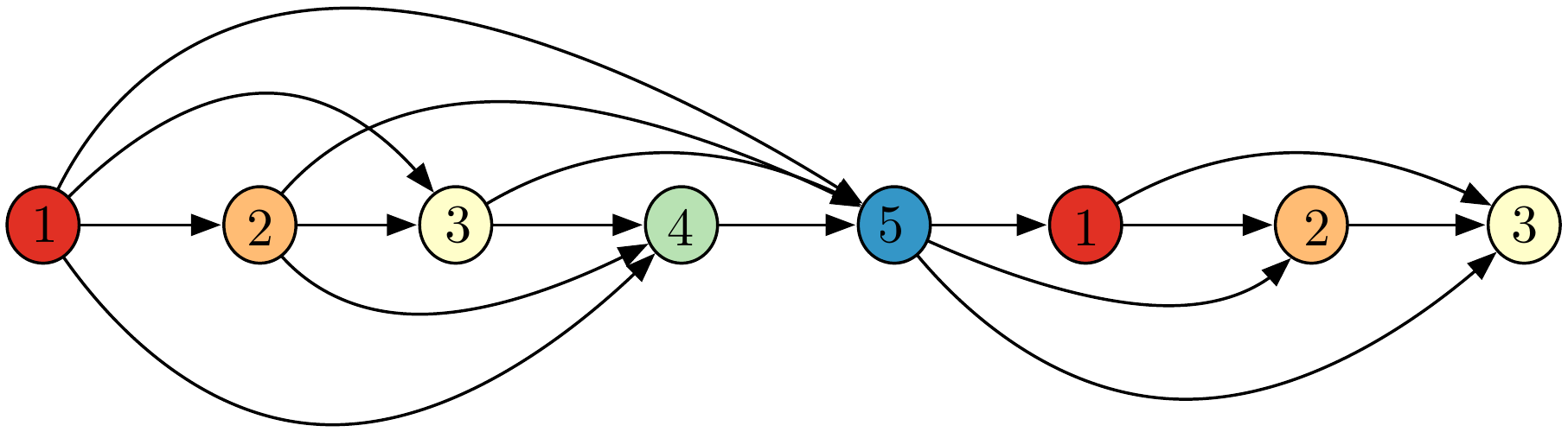}
	  \caption{A minimum coloring of the graph of Figure~\ref{fig:clique-example} using five colors. Using the lower bounds on the performance and Proposition~\ref{prop:coloring_upper}, we can conclude that $  \frac{1}{6}f(x^*_1,\ldots, x^*_n) <f(x_1,\ldots, x_n)<\frac{5}{8}f(x^*_1,\ldots, x^*_n)  $.}
	  \label{fig:chromatic-example}
\end{figure}

\subsection{Upper Bound via Universal Function}
 \label{sec:upper_univ}
 
The upper bound in the previous section required a submodular function $f$ that depended on the graph.  In the following, we propose a single function, independent of the graph topology, and analyze its performance for all graphs.    This allows us to state a few simple properties of graphs that limit its performance.  

The base set of the function is $X = \{e_1,\ldots,e_m\}$ where $m \geq n$ and we let $X_i = X$ for each agent $i$.  Given a choice $x_i \in X_i$ for each agent $i$, the value of the submodular function is
\begin{equation}
	\label{eq:universal_fn}
f(x_1,\ldots,x_n) = |\cup_i\{x_i\}|.
\end{equation}
That is, the value is given by the number of unique elements of $X$ chosen by the agents.  Clearly, an optimal solution is any one in which each agent chooses a different element from $X$, yielding a value of $n$.  Under the greedy algorithm, agent $i$ will compute its marginal reward as
\[
\Delta(e_j \;|\; X_{\mathrm{in}}(i)) = 
\begin{cases}
	1 & \text{if $e_j \notin X_{\mathrm{in}}(i)$} \\
	0 & \text{if $e_j \in X_{\mathrm{in}}(i)$}, \\
\end{cases}
\]
and thus agent $i$ will choose any strategy $e_j$ such that $e_j \notin X_{\mathrm{in}}(i)$. Suppose that each agent $i$ breaks ties by choosing the strategy with lowest index.  Then, we can write the greedy choice for agent $i$ as
\begin{equation}
	\label{eq:greedy_for_univ}
x_i = \min\{e_j \in X \;|\; e_j \notin X_{\mathrm{in}}(i) \}. 
\end{equation}

Next we relate the performance of the greedy algorithm on this function to properties of the underlying graph.  

\begin{proposition}
Consider the submodular function in~\eqref{eq:universal_fn} and any graph $\G = (V,E)$. If the greedy algorithm finds a solution within $k/n$ of the optimal for some $k > 0$ then each of the following properties hold:
\begin{enumerate}
	\item there is a vertex in $\G$ with in-degree of at least $k-1$;
	\item for each $\ell\in \{1,\ldots,k\}$  there are at least $\ell$ agents with in-degree of at least $k-\ell$;
	\item the number of edges in $\G$ is at least $k(k-1)/2$.
\end{enumerate}
\end{proposition}
\begin{proof}
By Equation~\eqref{eq:greedy_for_univ}, if $x_i = e_j$,  then $\{e_1,\ldots,e_{j-1}\} \subseteq X_{\mathrm{in}}(i)$.  From this we immediately see that for an agent to choose $e_k$, it must have an in-degree of at least $k-1$, proving property (i).  Property (ii) also follows by noticing that each agent that chooses a strategy with index $\geq j$ must have in-degree of at least $j-1$.  To achieve a value of $k$, each strategy from $e_1$ to $e_k$ must be chosen by an agent.
	
To see property (iii), we notice that the sum of in-degrees of all agents must be at least $\sum_{j=1}^k j = k(k-1)/2$. 
\end{proof}

Notice that the sequence of strategy choices in~\eqref{eq:greedy_for_univ} provides a simple and efficient algorithm for computing a performance upper bound for a given graph.  For completeness, we give the details in Algorithm~\ref{alg:upper_bound}.
\begin{algorithm}
	\caption{\textbf{Greedy Algorithm Upper Bound}}
	\label{alg:upper_bound}
\begin{algorithmic}[1]
  \REQUIRE A directed acyclic graph $\G = (V,E)$.
  \ENSURE An upper bound on the approximation ratio of greedy algorithm on $\G$.
  \STATE Topologically sort the vertices $V$
  \FOR{each $v\in V$}
  	\STATE Set $\texttt{value}[v] = 1$ 
  \ENDFOR
  \FOR{each $v\in V$ in topological order} \label{step:for_loop}
  	\STATE Set $\texttt{value}[v]$ to smallest integer $k$ such that for each $u \in \Nn(i)$, $\texttt{value}[u] \neq k$.  \label{step:smallest_int}
  \ENDFOR
  \RETURN $\frac{1}{|V|}\max_{v\in V} \texttt{value}[v]$
\end{algorithmic}
\end{algorithm}

\emph{Complexity of Algorithm~\ref{alg:upper_bound}:}  The complexity of Algorithm~\ref{alg:upper_bound} is $O(|V| + |E|)$.  The topological sort can be performed in $O(|V| + |E|)$ time.  The only detail is to implement line~\algostep{\ref{step:smallest_int}} in linear time, which essentially computes the smallest element not in an array.  This can be done using two passes through the array of in-neighbor values.  In the first pass, we populate a Boolean array of length $|\Nn(i)|$.  All entries of the array are initialized to \texttt{false}, and the $j$th entry is then reset to \texttt{true}  if and only if there is a vertex $u\in\Nn(i)$ with $\texttt{value}[u] = j$.  In the second pass, we scan the Boolean array for the first false entry.  This is the smallest value that is not chosen by an in-neighbor.  Thus, the total computation time for the for-loop in line~\algostep{\ref{step:for_loop}} is $O(|E|)$.

\subsection{Gap Between Adversarial and Universal Upper Bounds}

The strategies chosen by each agent for the submodular function in~\eqref{eq:universal_fn} provides a coloring of the graph $\G$.  That is,  the vertices $V_{\ell}:= \{i\in V\; | \; x_i = e_{\ell} \}$ are those colored with color $\ell$.  By construction, there are no edges between vertices of the same color.  This implies that the adversarial upper bound is tighter than the universal upper bound, which utilizes the minimum number of colors.  

A key question is how large the gap can be between the two upper bounds.  In general it can be arbitrarily large.  To see this, consider the following bipartite directed graph $\G = (V,E)$ consisting of $n = 2m$ vertices, where 
\[
V = \{u_1,\ldots,u_m\} \cup \{w_1,\ldots,w_m\}. 
\]
The graph contains all edges $(u_i,w_j)$ such that $i < j$ and all edges $(w_i, u_j)$ such that $i < j$ along with the edge $(u_m, w_m)$.  An example is shown in Figure~\ref{fig:2-color}.  We fix the topological ordering of the vertices to be $u_1,w_1,u_2,w_2,\ldots,u_m,w_m$.
\begin{figure}
	\includegraphics[width=0.27\linewidth]{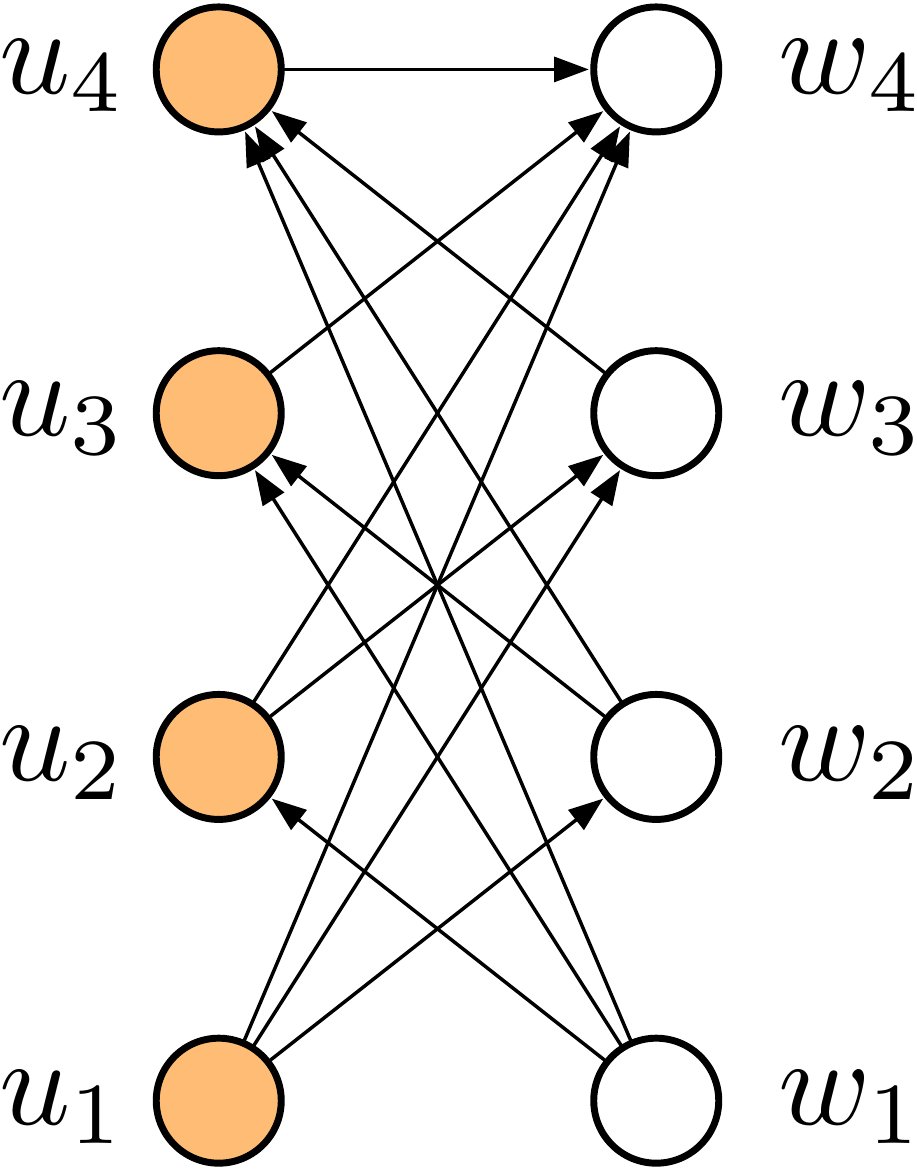}
	\centering
	\caption{A bipartite graph on eight vertices.  The graph is colored using just two colors, and its chromatic number is two. Under~\eqref{eq:greedy_for_univ}, vertex $w_4$ will select strategy $e_5$, and thus Algorithm~\ref{alg:upper_bound} will output $5/8$.}
	\label{fig:2-color}
\end{figure}
 For this graph, under Algorithm~\ref{alg:upper_bound}, the vertex $w_m$ chooses strategy $e_{m+1}$.  Thus
	\begin{enumerate}
		\item the chromatic number of $\G$ is 2, and thus the adversarial upper bound is $2/n$; and
		\item Algorithm~\ref{alg:upper_bound} returns $1/2 + 1/(2m)$.  
	\end{enumerate}

However, the advantage of Algorithm 1 is that it is a simple linear-time algorithm for computing an upper bound on achievable performance of a given graph topology.  In contrast, the adversarial function of Section~\ref{sec:adversarial} provides a much tighter upper bound, but requires solving an NP-hard problem.  An interesting connection to note is that Algorithm~\ref{alg:upper_bound} is essentially computing a greedy coloring of $\G$, where vertices are considered in topological order.  It is known that in general a greedy coloring does not provide a constant factor approximation to the minimum coloring~\cite{BK-JV:07}.

\begin{remark}[Gap between bounds for random graphs]
Given a probability $p \in[0,1]$ and a set $V$ of $n$ vertices, the Erd{\"o}s-R\'{e}nyi graph $\G(n,p)$ on $V$ is a random graph in which each edge $(i,j)$, where $i,j\in V$ and $i\neq j$ is included with probability $p$. From~\cite{bollobas1976cliques}, the chromatic number of almost every $\G(n,p)$ is 
\[
\left(\frac{1}{2} + o(1)\right)\log(1/(1-p))\frac{n}{\log n}.
\]
In addition, for almost every graph $\G(n,p)$, the greedy coloring uses at most twice as many colors~\cite{bollobas1976cliques}. Consider now random directed acyclic graphs generated as follows: we first generate an Erd{\"o}s-R\'{e}nyi graph, and then we use any correct method to assign directions to each edge such that the resulting directed graph becomes acyclic. Then, using the result above on the chromatic number of the underlying Erd{\"o}s-R\'{e}nyi graph,  Algorithm~\ref{alg:upper_bound} provides
\[
\left(\frac{1}{2} + o(1)\right)\log(1/(1-p))\frac{1}{\log n}
\]
as an upper bound on the performance.  In addition, for almost every graph $\G(n,p)$, the greedy coloring uses at most twice as many colors~\cite{bollobas1976cliques}, implying that for such graphs the gap between the two upper bounds is at most two.
\oprocend
\end{remark}

\subsection{Comparison of Lower and Upper Bounds}

In Section~\ref{sec:distributedgreedy}, we provided lower bounds on the performance of greedy algorithm~\eqref{eq:asyn-greedy-obs} for any monotone, normalized, submodular function.  This consisted of results for several graph topologies, and a general result based on the clique number of the graph.  Table~\ref{tab:lower_upper} compares these lower bounds with the upper bounds obtained from Proposition~\ref{prop:coloring_upper} and Algorithm~\ref{alg:upper_bound}.  Given a graph topology, the lower bound provides a minimum performance guarantee for all submodular functions.  In constrast, the upper bounds provide limitations on performance for a specific (worst-case) submodular function.  
\begin{table}[H]
	\centering
	\ra{1.2}
	\begin{tabular}{@{} l ccc @{}}
		\toprule
		Graph & Lower & $\chi(\G)$ Upper & Alg.~\ref{alg:upper_bound} Upper \\
		\midrule 
		Empty & $1/n$ & $1/n$ & $1/n$ \\
		Complete acyclic & $1/2$ & 1 & 1 \\
		Interconnected cliques & $1/(2\kappa)$ & $1/\kappa + 1/n$ & $1/\kappa + 1/n$ \\
		General DAG & $\frac{1}{n - \omega(\G) + 2}$ & $\chi(\G)/n$ &  Alg.~\ref{alg:upper_bound} \\
		\bottomrule
	\end{tabular}
	\caption{Comparison between the lower bound in Section~\ref{sec:distributedgreedy} and the two upper bounds for three graph topologies. }
	\label{tab:lower_upper}
\end{table}
The entires in Table~\ref{tab:lower_upper} for the interconnected cliques graph assume $\kappa$ cliques, each of size $n/\kappa$.

\section{Simulation Results}\label{section:coverage}
In this section we explore the relation between the lower bound, the two upper bounds, and the true performance of the greedy algorithm on several types of randomly generated graphs. The classes of random graphs chosen here are Erd{\"o}s-R\'{e}nyi, Barab{\'a}si-Albert preferential attachment, and Watts- Strogatz model graphs, which are commonly used to model social networks~\cite{NB-GB-SZ:16}, as well as physical networks~\cite{soltan2015analysis}. Given a graph $\G$, we compute the chromatic number $\chi(\G)$ and the and clique number $\omega(\G)$ using their standard integer programming formulations~\cite{pattillo2011clique}, and solve each integer program using Gurobi~\cite{gurobi}.
\begin{figure}[htb!]
	\centering
	\includegraphics[width=\linewidth]{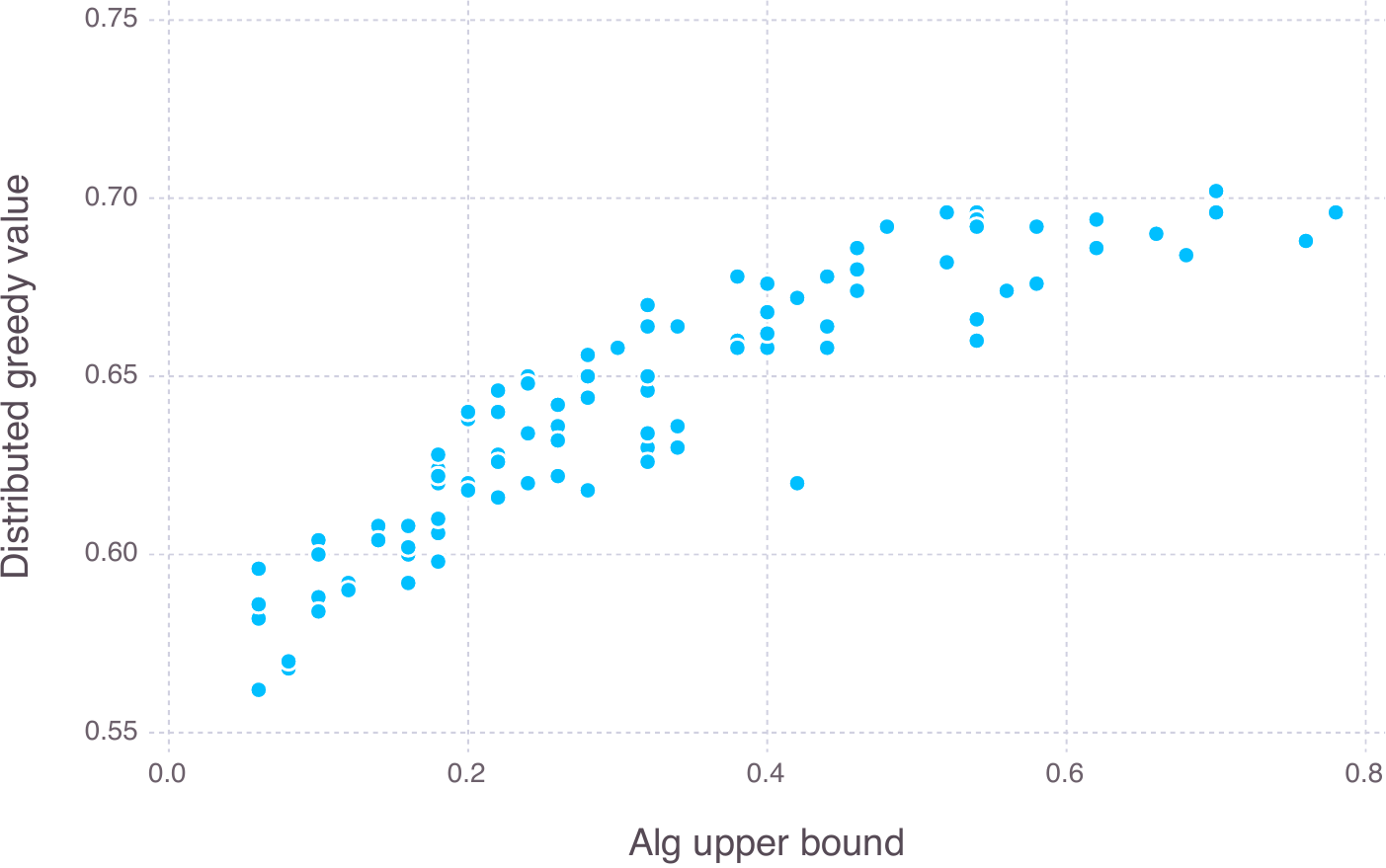}
	\caption{Comparison between the upper bound from Algorithm 1 and the performance on the distributed greedy algorithm over $100$ random graphs. The value of the upper bound correlates strongly with the true performance of the distributed greedy algorithm (Spearman's rank correlation coefficient of $0.92$).}
	\label{fig:distributed_upper}
\end{figure}

First, we explored the relation between the true performance of the distributed greedy algorithm and the upper bound given by Algorithm~\ref{alg:upper_bound}.  We generated a submodular function using the disk coverage problem illustrated in Figure~\ref{fig:disk_selection} with 50 agents.  Each agent has the choice of three disks each of radius $r=0.07$, uniformly randomly distributed in the unit square.  We then generated 100 random graphs describing the information structure between the agents.  Each graph is generated as follows.  We select a probability $p$ uniformly randomly from $[0,1]$, generate a directed Erdos-Renyi graph with probability $p$, fix a random ordering of the vertices, and delete all edges that are directed to an earlier vertex in the ordering, creating a random DAG.  Figure~\ref{fig:distributed_upper} displays the results on these 100 graphs.  The upper bound, which is an approximation ratio relative to the global optimal, and the value obtained by the distributed greedy algorithm, which is the fraction of the unit square covered by the selected disks, are highly correlated.

Next, Figure~\ref{fig:preferential_attachment} compares the lower bound to the two upper bounds for graphs generated according to the Barab{\'a}si-Albert preferential attachment model~\cite{barabasi1999emergence} a function of the number of vertices in the graph (i.e., the number of agents in the network).  Each graph is created by starting with a complete graph on five vertices, and then adding vertices using preferential attachment.  An acyclic graph is then created by generating a random ordering of the vertices and deleting cyclic edges as described above.  Each data point shows the mean of 20 preferential attachment graphs, where error bars give the standard deviation.
\begin{figure}[htb!]
	\centering
	\includegraphics[width=\linewidth]{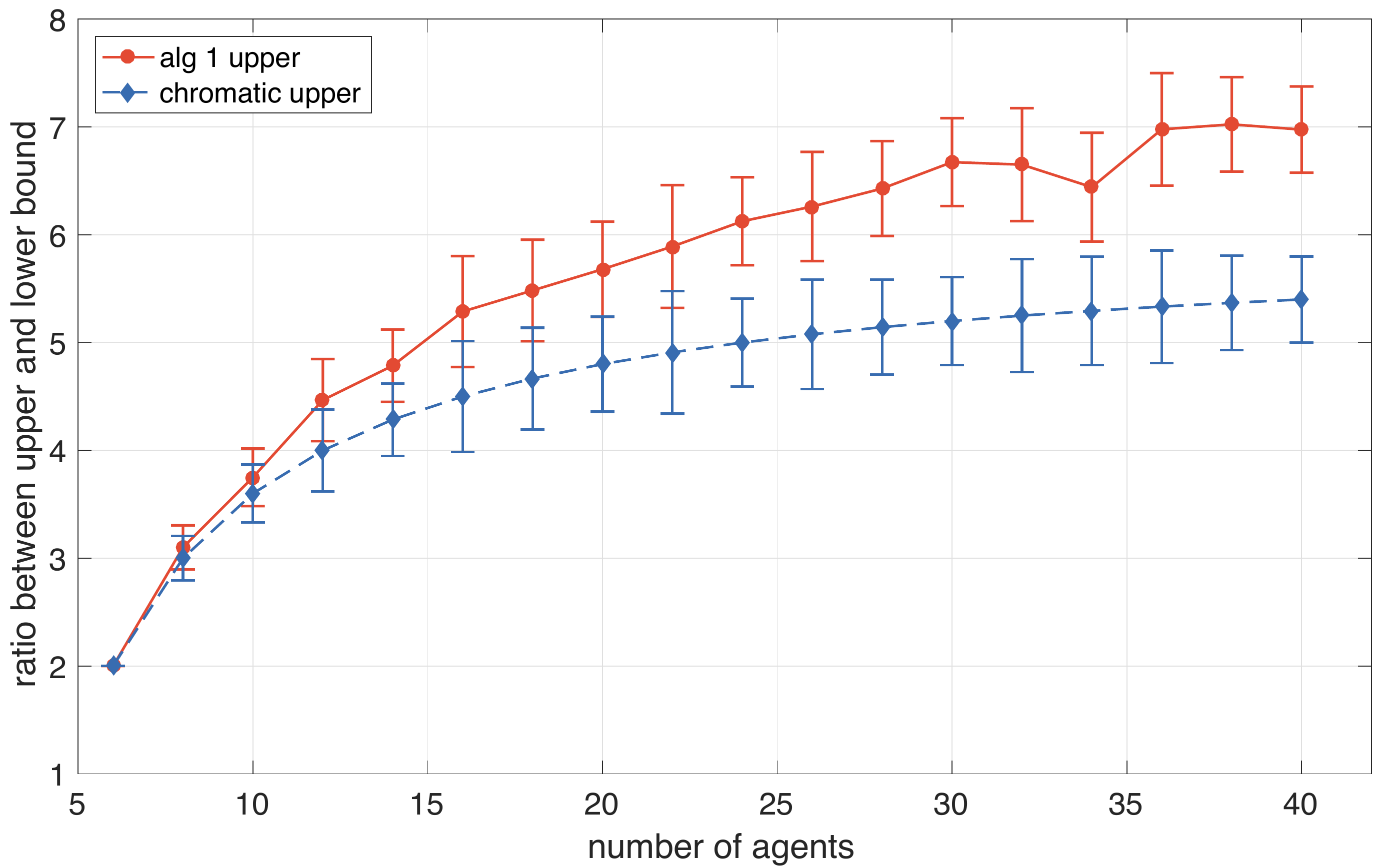}
	\caption{Ratio between upper and lower bounds as a function of the number of agents for networks formed via preferential attachment.  
	Each graph is grown from a complete directed graph on 5 nodes, and each new node is connected to five neighbors using the preferential attachment rule.  Each data point shows the mean of 20 trials, and error bars give the standard deviation.}
	\label{fig:preferential_attachment}
\end{figure}
The figure gives insight into the gap between lower and upper bounds, and the growth of this gap with the number of agents.  Note, this plot shows the dependence of the bounds on the information structure, and holds for all functions $f$.

Finally, Figure~\ref{fig:small-world} compares the lower bound to the two upper bounds for small-world graphs created according to the Watts-Strogatz model~\cite{watts1998collective}.  In this experiment, each graph begins as a ring lattice over 25 vertices, where each vertex is connected to $K$ neighbors on each side.  The value of $K$ is plotted on the horizontal axis in Figure~\ref{fig:small-world}.  Then, each edge is rewired with a probability $\beta = 0.25$, by shifting one of its endpoints to a new vertex at random, while avoiding self-loops and duplication of edges.  A DAG is then created by fixing a random ordering of the vertices, and removing acyclic edges, as described above.
\begin{figure}[htb!]
	\centering
	\includegraphics[width=\linewidth]{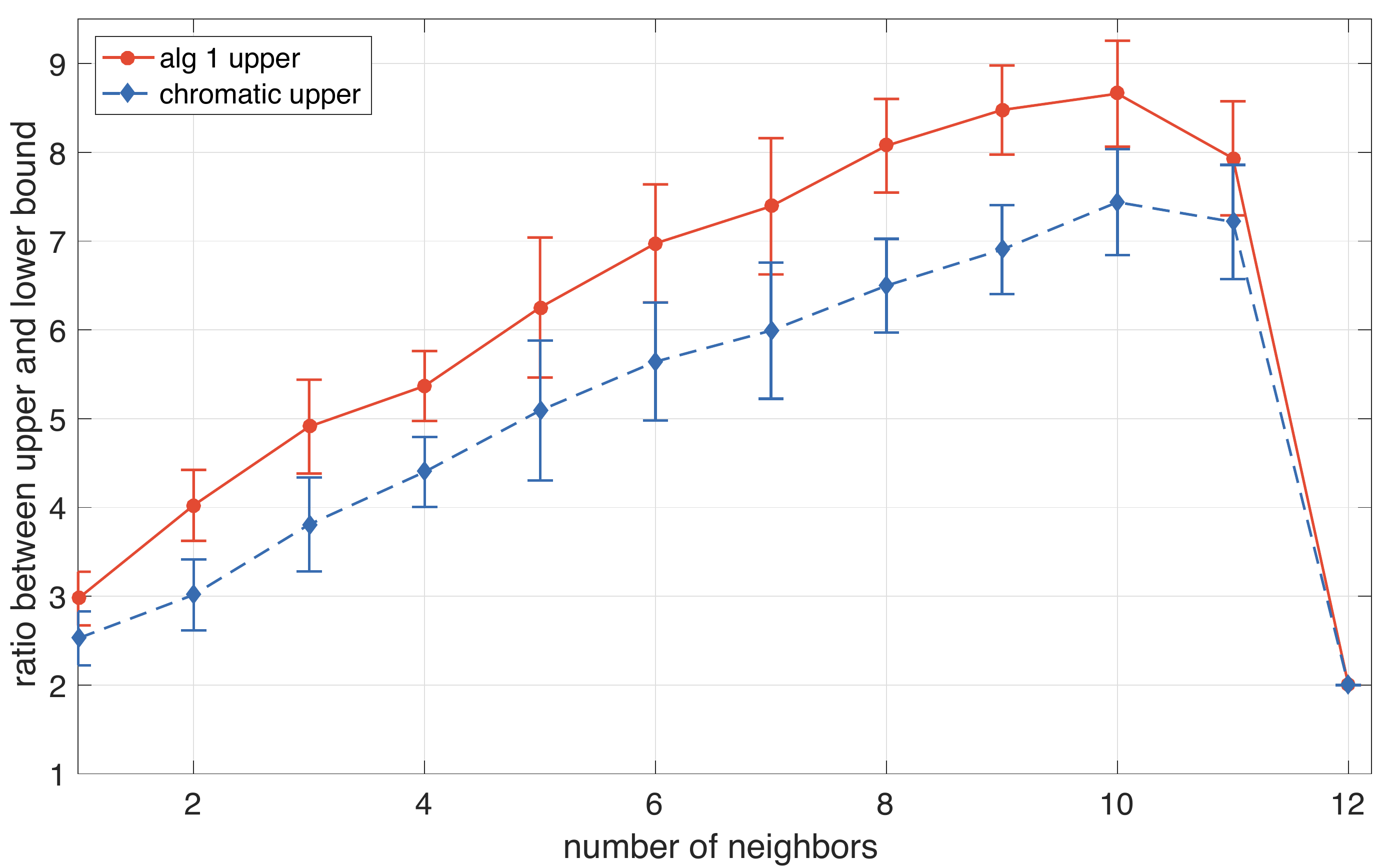}
	\caption{Ratio between upper and lower bounds as a function of the number of neighbors for each agent prior to the rewiring operation.  Each graph contains $25$ agents, and the rewiring probability $\beta = 0.25$.  Each data point shows the mean of $20$ trials, and error bars give the standard deviation.}
	\label{fig:small-world}
\end{figure}

Each data point shows the mean of 20 small-world graphs, where error bars give the standard deviation.  As the number of neighbors $K$ increases, the graphs become more connected.  When $K=12$, each vertex is connected to all others, and thus the resulting DAG is complete.  Of note in this figure is the relatively small gap between the two upper bounds, which gives support to the use of Algorithm~\ref{alg:upper_bound} as an efficient approximation to the chromatic number bound.  Also of note is the large change in the ratio of the bounds as $K$ is increased from 11 to 12.  Recall that for a complete DAG (i.e., $K=12$), the lower bound is $1/2$, while the upper bounds are both $1$.  On the other hand, as discussed immediately after Corollary~\ref{cor:clique_bound}, removing even just a single edge from a complete DAG has a large impact on the lower bound, which is observed in the large change in the ratio.

\section{Conclusions and future work}
We have studied a class of submodular maximization problems under matroid constraint, where the strategy set is partitioned into private strategy sets, assigned to a group of agents; in this sense, the maximization task is distributed among them. The agents take actions sequentially and have limited information about the actions taken by agents prior to them. We have investigated the limitations that the lack of information about the actions of other agents can impose by characterizing upper and lower bounds  on the performance of local greedy algorithms. Our lower bound  depend on the clique number of a graph that captures the information structure, whereas our upper bounds are given in terms of the chromatic number of the underlying graph. Studying scenarios where agents are allowed to modify their choices in multiple rounds, obtaining tighter upper bounds for more homogeneous classes of submodular functions, investigating under what additional assumptions on the submodular functions one can generalize the lower bounds provided, and exploring other applications of distributed submodular optimization are interesting avenues of future research.

 \section*{Acknowledgments}
The authors wish to thank Mohd. Shabbir Ali and Professor Jason Marden for their useful comments on an earlier version of this paper. We also thank the associate editor and the reviewers for their helpful comments.

\end{document}